\documentclass[lettersize,onecolumn]{IEEEtran}[12pt]

\usepackage{geometry}
\usepackage{amsmath,amsfonts}
\usepackage{algorithmic}
\usepackage{algorithm}
\usepackage{array}
\usepackage[caption=false,font=normalsize,labelfont=sf,textfont=sf]{subfig}
\usepackage{textcomp}
\usepackage{stfloats}
\usepackage{url}
\usepackage{verbatim}
\usepackage{graphicx}
\usepackage{cite}
\usepackage{tikz}
\usepackage{float}
\usepackage{setspace}
\usetikzlibrary{positioning}
\usetikzlibrary{calc} 
\usepackage{caption}
\usepackage{dsfont}

\newcommand{\set}[1]{\mathcal{#1}}

\RequirePackage{amsmath}
\RequirePackage{amsopn}
\RequirePackage{amssymb}
\RequirePackage{amsthm}
\newtheorem{theorem}{Theorem}
\begin{document}

\onehalfspacing

\title{The Equivalence of Causal and Noncausal State Information on
  Bipartite Networks With State-Cognizant Receivers}

\author{Amos Lapidoth, Baohua Ni, and Ligong Wang%
    \thanks{The authors are with the Department of Information Technology and Electrical Engineering, ETH Zurich, 8092 Zurich, Switzerland (email: \{lapidoth, baohni, ligwang\}@isi.ee.ethz.ch).}
}
        % <-this % stops a space

%\markboth{Journal of \LaTeX\ Class Files,~Vol.~1, No.~2, December~2023}%

%\IEEEpubid{0000--0000~\copyright~2023 IEEE}

\maketitle

\begin{abstract}
  State-dependent bipartite networks with state-cognizant receivers
  and state-informed transmitters are studied. Such networks have no
  nodes that both transmit and receive. Examples are the multi-access
  channel, the broadcast channel, and the interference
  channel. Without computing the capacity region of the network, it is
  shown that if the state sequence is ergodic and autonomous, and if,
  conditionally on the state sequence, the network law is memoryless,
  then the network capacity region does not depend on whether the
  state information is provided to the encoders causally or
  noncausally.
%
  % We examine the capacity of state-dependent bipartite networks where
  % receivers are cognizant of the state sequence. Traditionally,
  % proving the equivalence of causal and noncausal state information
  % has usually required the explicit computation of capacity
  % expressions, which are currently unknown for many complex
  % configurations. We present a proof technique that bypasses this
  % requirement to demonstrate that the causal and noncausal capacity
  % regions ($\set{C}_\text{c}$ and $\set{C}_\text{nc}$) are identical
  % for any network where nodes are strictly partitioned into disjoint
  % sets of transmitters and receivers. This result encompasses the
  % Multiple Access Channel, Broadcast Channel, and Interference
  % Channel. Furthermore, we show that this equivalence holds beyond
  % independent and identically distributed state sequences, requiring
  % only that the state sequence is autonomous and satisfies the weak
  % law of large numbers.
\end{abstract}
\section{Introduction}
The study of state-dependent channels began shortly after the
inception of Information Theory. Early work by Shannon
\cite{Shannon1958} established the single-user capacity for the causal
case, where the state sequence is independent and identically
distributed (IID), and where the transmitter is cognizant of the
past-and-present states. He showed that the capacity is achieved by
what we now call ``Shannon strategies.''  The case where the
transmitter is informed of the states noncausally was solved by
Gel'fand and Pinsker \cite{gelfand1980}. More involved results and a
detailed historical progression of channel coding with states can be
found in the survey paper \cite{keshet2008survey}. See also
\cite{YosseyBC} for some results on a class of state-dependent
broadcast channels in which the channel law from the strong receiver to
the weak receiver does not depend on the state. (For this class,
Steinberg fully characterized the capacity region in two cases: when
the non-causal state information is provided to both the encoder and
the nondegraded decoder, and when causal state information is
available only to the encoder.)

  For the single-user channel with a state-cognizant receiver, causal
  and noncausal state information at the transmitter are equally
  beneficial (in terms of capacity); see, e.g., \cite[Theorem
  1]{Jafar2006}. This was shown to also apply to some 
  multi-access channels (namely, those of a double state of
  independent components) \cite[Theorem~5]{Jafar2006}. Such results
  are usually proved by computing the two capacities and showing that
  they coincide.

  But what about more intricate networks, such as the
Interference Channel or the general Broadcast Channel, where the two capacity
regions are not known? Are the two regions still equal? This question
is answered here in the affirmative using a proof technique that
bypasses the need for explicit expressions for the two capacities.  
%In
%fact, we show that \emph{Shannon strategies}---where the encoders
%ignore the past states $S_{1}^{i-1}$ and the future states
%$S_{i+1}^{n}$ when producing the time-$i$ symbol---can 
%achieve the capacity region, irrespective of whether the state
%information is provided to the transmitters causally or noncausally.
This result holds
% conclusion holds true
for all state-dependent \emph{bipartite} networks (e.g. the Multiple
Access Channel, the Broadcast Channel, or the Interference Channel,
but not the Relay Channel or the Two-Way Channel),
in which all nodes are either transmitters or receivers (but never
both).  In fact---provided that the state sequence is unaffected by
the channel inputs---the
% The technique works not only when
state sequence need not be IID; it suffices that it be ergodic, or at least that it satisfy the
% , but whenever it satisfies the
weak law of large numbers. % (and is not affected by the channel inputs).
We do, however, assume that, conditional on the state sequence, the
network is memoryless.

\section{Problem Setup}
We consider a state-dependent network of multiple nodes that are
divided into two disjoint sets of transmitters and receivers as in
Fig.~\ref{fig:channel}. We thus exclude the
Relay Channel (because the relay is neither a transmitter nor a
receiver) and the Two-Way Channel (where the nodes both
transmit and receive).
We denote the number of transmitters $k$ and the number of receivers
$\ell$. The channel is characterized by its transition law
$W(y_1,\ldots, y_\ell\,|\,x_1,\ldots,x_k,s)$, which is the
conditional probability mass function (PMF) of the outputs observed by the
receivers given the symbols sent by the transmitters and the
state. The state, inputs, and outputs all take values in finite sets $\set{S}$,
$\set{X}_1,\ldots,\set{X}_k$, and $\set{Y}_1,\ldots,\set{Y}_\ell$.

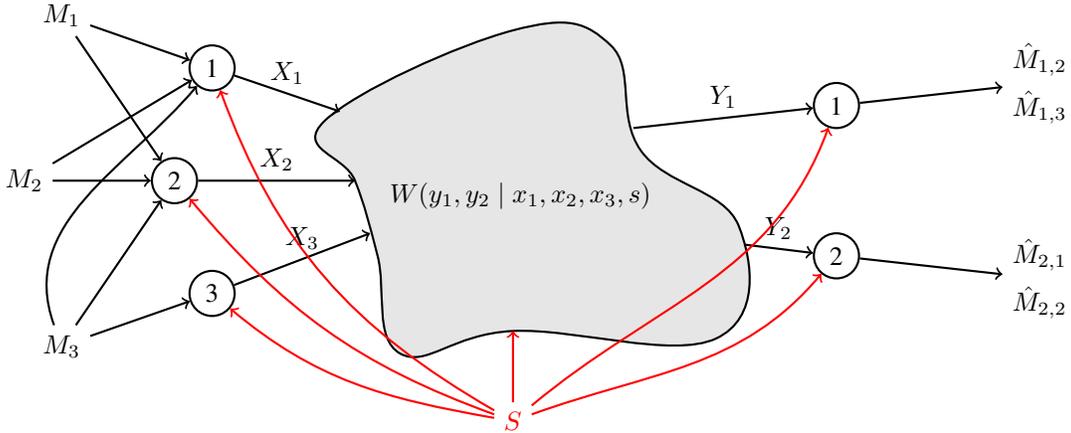
\begin{figure}[tbp]
\centering    
\begin{tikzpicture}[->, thick]

\path[draw, thick, fill=gray!20] plot[smooth cycle, tension=0.7] coordinates {
        (0, 2.0)      % Top center
        (1.3, 1.8)    % Top right bump
        (1.7, 0.4)    % Right top edge
        (3, -0.6)   % Right bottom edge
        (2.7, -2.1)   % Bottom right bump
        (0, -2)     % Bottom center indent
        (-1.5, -2.3)  % Bottom left bump
        (-1.8, -1)  % Left bottom edge
        (-2.1, 0)   % Left top edge
        (-2.5, 0.8)   % Top left bump
    };
    \node at (0.1,-0.2) {$W(y_1, y_2 \mid x_1,x_2,x_3,s)$};
    
    \node (state) at (0, -3.2) [red] {$S$};
    \draw[->,red] (state) -- (0, -2);
    
    \node[circle, draw, minimum size=6mm, inner sep=0pt] (E1) at (-4, 1.5) {1};
    \node[circle, draw, minimum size=6mm, inner sep=0pt] (E2) at (-4.5, 0) {2};
    \node[circle, draw, minimum size=6mm, inner sep=0pt] (E3) at (-4, -1.5) {3};

    \node[circle, draw, minimum size=6mm, inner sep=0pt] (D1) at (4.3, 1) {1};
    \node[circle, draw, minimum size=6mm, inner sep=0pt] (D2) at (4.3, -1) {2};
    
    \node (M1)  at (-6, 2.2)  {$M_1$};
    \node (M2)   at (-6.5, 0)    {$M_2$};
    \node (M3) at (-6, -2.2) {$M_3$};

     \node[align=left] (out1) at (7, 1.3) {$\hat{M}_{1,2}$\\[2mm]$\hat{M}_{1,3}$};
    \node[align=left] (out2) at (7, -1.3) {$\hat{M}_{2,1}$\\[2mm]$\hat{M}_{2,2}$};

    \draw[->] (M1) -- (E1);
    \draw[->] (M1) -- (E2);
    
    \draw[->] (M2) -- (E1);
    \draw[->] (M2)  -- (E2);
    
    \draw[->]  (M3) to [out=110,in=230] (E1);
    \draw[->] (M3) -- (E2);
    \draw[->] (M3) -- (E3);

    \draw[->] (E1) -- (-2.3, 0.92) node[midway, above] {$X_1$};
    \draw[->] (E2) -- (-2.1, 0) node[midway, above] {$X_2$};
    \draw[->] (E3) -- (-1.9, -0.7) node[midway, above] {$X_3$};

    \draw[->] (1.6, 0.7) -- (D1) node[midway, above] {$Y_1$};
    \draw[->] (3.07, -0.85) -- (D2) node[midway, above] {$Y_2$};

    \draw[->] (D1) -- (out1);
    \draw[->] (D2) -- (out2);

    \draw[->,red]  (state) to [out=150,in=-70] (E1);
    \draw[->,red]  (state) to [out=160,in=-50] (E2);
    \draw[->,red]  (state) to [out=170,in=-40] (E3);

    \draw[->,red]  (state) to [out=40,in=-110] (D1);
    \draw[->,red]  (state) to [out=20,in=-130] (D2);
\end{tikzpicture}
    
\caption{A Multiterminal Network with three transmitters, two receivers, and three messages.}
\label{fig:channel}
\end{figure}

The state sequence is assumed to satisfy the weak law of large
numbers: there exists a PMF $P_S$ of support~$\set{S}$
such that, for every $s\in\set{S}$ and every $\epsilon>0$,
\begin{equation}\label{eq:WLLN}
\lim_{n\to\infty} \Pr\left\{ \bigg|\frac{1}{n}\sum_{i=1}^n \mathds{1}\{S_i=s\} - P_S(s)\bigg| > \epsilon \right\} = 0.
%        P_S(s) \quad\textnormal{almost surely.}
\end{equation}
This is the case if the state process is ergodic, e.g., if it is an irreducible finite-state Markov process \cite{durrett2019probability}. We assume that the state sequence is autonomous, i.e., that it is not influenced by the channel inputs.
Conditional on the state sequence, the network
is memoryless:
%
%The joint probability distribution of the outputs and states, given
%the inputs and the initial state, factors as follows:
\begin{IEEEeqnarray}{rCl}
     \IEEEeqnarraymulticol{3}{l}{ 
     \Pr \Big\{ Y_{1,1}^n = y_{1,1}^n,\ldots,\, Y_{\ell,1}^n =
     y_{\ell,1}^n \,\Big|\, X_{1,1}^n =
     x_{1,1}^n,\ldots,\, X_{k,1}^n = x_{k,1}^n, \,S_{n} = s^{n}
     \Big\}
     }\nonumber\\* \qquad\qquad\qquad\qquad\qquad\qquad\qquad\qquad\qquad
     &= &\prod_{i=1}^n
     W(y_{1,i},\ldots, y_{\ell,i}\,|\,x_{1,i},\ldots, x_{k,i},s_i). %\nonumber\\*\ 
\end{IEEEeqnarray}

Consider a set of $\Sigma$ messages $\{M_\sigma\}_{\sigma\in[1:\Sigma]}$, where each message is presented to at least one transmitter and is intended for at least one receiver. Message $M_\sigma$ is of rate $R_\sigma$ and takes values in the set $\set{M}_\sigma \triangleq [1:2^{nR_\sigma}]$. The messages presented to Encoder~$a$ are
$\{M_\sigma\}_{\sigma \in \set{I}_a}$, and those intended for Decoder~$b$ are $\{M_\sigma\}_{\sigma \in \set{J}_b}$, where $\set{I}_a,\set{J}_b \subseteq [1:\Sigma]$. We refer to $\bigl(R_\sigma\bigr)_{\sigma\in[1:\Sigma]}$ as the ``rate vector.''  %Here, some messages may be presented to only one encoder, while others
Given $b\in[1:\ell]$ and $\sigma\in\set{J}_b$, we use $\hat{M}_{b,\sigma}$ to denote Decoder~$b$'s guess of   $M_{\sigma}$. Since Decoder~$b$ is cognizant of the states, its decoding function has the form
\begin{IEEEeqnarray}{c}\label{eq:dec}
    \phi_b\colon \set{Y}_b^{\times n} \times \set{S}^{\times n} \to \bigotimes_{\sigma\in\set{J}_b}\set{M}_\sigma,\quad \left(y_b^n,s^n\right)\mapsto \big\{\hat{M}_{b,\sigma}\big\}_{\sigma\in\set{J}_b}.% = \phi_b\Bigl(\Bigr)
\end{IEEEeqnarray}

As for the encoders, we consider both causal and noncausal state
information. In the causal case, the Time-$i$ symbol produced by
Encoder~$a$ is determined by the messages presented to it, namely,
$\{M_{\sigma}\}_{\sigma \in \set{I}_a}$, and the state sequence up to Time~$i$, namely,
$s^i$. Encoder~$a$ is thus specified by $n$ functions
$\{f_{a,i}^\textnormal{c}\}_{i=1}^n$, where
\begin{IEEEeqnarray}{c}\label{eq:causal_enc}
    f_{a,i}^{\textnormal{c}}\colon \left( \bigotimes_{\sigma\in \set{I}_a} \set{M}_\sigma \right) \times \set{S}^{\times i} \to \set{X}_a,\quad
    \left(\left\{M_\sigma\right\}_{\sigma\in\set{I}_a},s^i\right) \mapsto x_{a,i}.
\end{IEEEeqnarray}
%where the encoder takes as input all the messages $M_J$ known to it ($a\in J$).

In the noncausal case, every encoder is cognizant of the entire state sequence prior to transmission, hence the encoding function employed by Encoder~$a$
has the form
\begin{IEEEeqnarray}{c}\label{eq:noncau_enc}
    f_{a}^{\textnormal{nc}} \colon \left( \bigotimes_{\sigma\in \set{I}_a} \set{M}_\sigma \right) \times \set{S}^{\times n} \to \set{X}_a^{\times n},\quad
   \left(\left\{M_\sigma\right\}_{\sigma\in\set{I}_a},s^n\right) \mapsto x_a^n.
\end{IEEEeqnarray}
We sometimes also write it as $n$ functions $\{f_{a,i}^\textnormal{nc}\}_{i=1}^n$,
\begin{equation}\label{eq:nc_enc_separate}
f_{a,i}^\textnormal{nc} \colon \left( \bigotimes_{\sigma\in \set{I}_a} \set{M}_\sigma \right) \times \set{S}^{\times n} \to \set{X}_a,\quad
    \left(\left\{M_\sigma\right\}_{\sigma\in\set{I}_a},s^n\right) \mapsto x_{a,i}.
\end{equation}

%We shall also consider scenarios where the encoders employ Shannon
%strategies, which, when producing the time-$i$ channel input, ignore
%the past states $S_{1}^{i-1}$ and the future states $S_{i+1}^{n}$ and
%only take into account the message/s and the time-$i$ state. When
%restricted to using Shannon strategies, Encoder~$a$ is characterized
%by $n$ mappings of the form
%\begin{IEEEeqnarray}{c}\label{eq:Shannon}
%    f_{a,i}^{\textnormal{Sh}}\colon \left( \bigotimes_{\sigma\in \set{I}_a} \set{M}_\sigma \right) \times \set{S} \to \set{X}_a,\quad
%    \left(\left(M_\sigma\right)_{\sigma\in\set{I}_a},s_i\right) \mapsto x_{a,i}.
%\end{IEEEeqnarray}

In both the causal and the noncausal cases, 
%the event that Decoder~$b$ errs in decoding $M_{\sigma}$ is denoted
%\begin{IEEEeqnarray}{c}
%    \set{E}_{b,\sigma} \triangleq \bigl\{\hat{M}_{b,\sigma}\neq M_{\sigma}\bigr\},
%\end{IEEEeqnarray}
%which means that the reconstruction of the message $M_J$ by
%Decoder~$b$ is wrong. Thus,
the error event $\set{E}$ is defined as
\begin{equation}
\mathcal{E} = \left\{ \exists b,\sigma\colon \hat{M}_{b,\sigma}\neq M_{\sigma} \right\}.
\end{equation}
The average probability of error $\Pr\{ \mathcal{E} \}$ is
%\begin{IEEEeqnarray}{c}
%    P_{\text{e}}^{(n)}\triangleq \Pr \left\{\right\}
%\end{IEEEeqnarray}
computed with all messages being uniformly distributed and mutually independent. (We sometimes write $\mathrm{Pr}_{\textnormal{c}}\{\set{E}\}$ for the causal setting and $\mathrm{Pr}_{\textnormal{nc}}\{\set{E}\}$ for noncausal.)

A rate vector $\left(R_\sigma\right)_{\sigma\in[1:\Sigma]}$ is said to be achievable if there exists a sequence of encoding and decoding functions, indexed by the blocklength~$n$, such that $\Pr\{ \mathcal{E} \}$ tends to 0 as $n$ tends to infinity. The capacity regions $\set{C}_\textnormal{c}$ and $\set{C}_\textnormal{nc}$ in the causal and noncausal settings are defined to be the closure of the sets of achievable rate vectors in their respective settings.

\section{Main Result and Proof}
\begin{theorem}\label{thm}
The capacity region of a bipartite network, where all receivers are state-cognizant, is the same in the causal and noncausal settings:
    \begin{IEEEeqnarray}{c}
    \label{eq:main}
        \set{C}_\textnormal{c} = \set{C}_\textnormal{nc}.
    \end{IEEEeqnarray}
\end{theorem}

\medskip

\begin{IEEEproof}
  It suffices to show that
  $\set{C}_\textnormal{nc}\subseteq \set{C}_\textnormal{c}$, as
  the reverse inclusion holds because every causal encoding strategy
  can also be employed with noncausal state information. Let the rates
  $\left(R_\sigma\right)_{\sigma\in[1:\Sigma]}$ be in the interior of
  $\set{C}_\textnormal{nc}$ and hence, given any $p>0$, for
  sufficiently large~$n$, there exist \emph{noncausal} encoding
  functions~\eqref{eq:noncau_enc} and decoding
  functions~\eqref{eq:dec} of these rates with
  $\mathrm{Pr}_{\textnormal{nc}}\{\mathcal{E}\} \le p$. We will show
  that, given any $\delta>0$, the rate vector
  $\left(R_\sigma/(1+2\delta) \right)_{\sigma\in[1:\Sigma]}$ is
  achievable with causal state information, i.e., that, for this rate
  vector, for sufficiently large~$n$, there exist \emph{causal}
  encoding functions of the form \eqref{eq:causal_enc}, together with
  corresponding decoding functions
%  for the rate vector
%  $\left(R_\sigma/(1+2\delta) \right)_{\sigma\in[1:\Sigma]}$
  satisfying $\mathrm{Pr}_{\textnormal{c}}\{\mathcal{E}\}\le 3p$. The
  claim \eqref{eq:main} will then follow.

By \eqref{eq:WLLN}, for sufficiently large $n$, 
        \begin{IEEEeqnarray}{c} \label{eq:PrS12}
            \Pr\left\{S^n\in \set{T}^{(n)}_{\delta}(P_S)\right\} > \frac{1}{2},
        \end{IEEEeqnarray}
        where $\set{T}^{(n)}_{\delta}(P_S)$ denotes the
        $\delta$-strongly typical set with respect to $P_S$
        \cite{csiszartextbook}.
%    \begin{IEEEproof}[Formal Proof for Theorem \ref{thm}]
        Expressing $\mathrm{Pr}_{\textnormal{nc}}(\mathcal{E})$ as
        \begin{IEEEeqnarray}{c}\label{eq:PEtotal}
        \mathrm{Pr}_{\textnormal{nc}}\left\{\mathcal{E}\right\} = \sum_{s^n\in \set{S}^{\times n}}  \mathrm{Pr}_{\textnormal{nc}}\left\{\mathcal{E}\,\middle|\, S^n=s^n \right\}\cdot \Pr\left\{S^n=s^n\right\} 
    \end{IEEEeqnarray}
    demonstrates that the inequalities
    $\mathrm{Pr}_{\textnormal{nc}}\{\mathcal{E}\} \le p$ and
    \eqref{eq:PrS12} imply the existence of some length-$n$ state
    sequence $\tilde{s}^n$ for which the following two conditions hold:
%        By assumption, there exists a noncausal coding strategy of desired rates for which
%    \begin{IEEEeqnarray}{c}\label{eq:PEtotal}
%        \mathrm{Pr}_{\textnormal{nc}}\left\{\mathcal{E}\right\} = \sum_{s^n\in \set{S}^{\times n}}  \mathrm{Pr}_{\textn%ormal{nc}}\left\{\mathcal{E}\,\middle|\, S^n=s^n \right\}\cdot \Pr\left\{S^n=s^n\right\} < p.
%    \end{IEEEeqnarray}
%    This and \eqref{eq:PrS12} together guarantee that there exists
    \begin{subequations}
      \label{sub:Choicetildes}
    \begin{equation}
    \tilde{s}^n\in \set{T}^{(n)}_{\delta}(P_S)
    \end{equation}
    and, for the given noncausal coding scheme,
    \begin{equation}\label{eq:ChoiceOftilde}
     \mathrm{Pr}_{\textnormal{nc}}\left\{\mathcal{E}\,\middle|\, S^n=\tilde{s}^n \right\}<2p.
    \end{equation}
    \end{subequations}
    Our causal coding scheme depends highly on $\tilde{s}^{n}$, so
    it is crucial that the encoders and decoders agree on it ahead of
    time. Henceforth, it will be fixed.
    % Fix such an $\tilde{s}^n$ for the rest of this proof, and let

    Let $\tilde{x}_{a,i}$ denote the time-$i$ symbol Encoder~$a$
    produces in the noncausal case when it wishes to convey the
    given messages after the state sequence $\tilde{s}^{n}$ has been
    revealed to it noncausally. (It is thus the result of applying the
    mapping $f_{a,i}^\textnormal{nc}$ of~\eqref{eq:nc_enc_separate} to
    the given messages and the sequence $\tilde{s}^{n}$.)

    Define the longer blocklength
    \begin{IEEEeqnarray}{rCl}
      \bar{n} & = & (1+2\delta)n.
    \end{IEEEeqnarray}
    We next describe the blocklength-$\bar{n}$ causal coding scheme
    that we propose in order to convey the messages
    $\{M_\sigma\}_{\sigma\in[1:\Sigma]}$ (that the noncausal scheme
    conveys in $n$ channel uses) when the transmitters are provided
    the state sequence $s^{\bar{n}}$ causally.
    %Defining we proceed to
    %show how to encode said messages when an arbitrary
    %length-$\bar{n}$ state sequence $s^{\bar{n}}$ is presented to the
    %encoders causally.
    The time-$i$ symbol produced by Encoder~$a$ of our proposed causal
    scheme will be denoted $x_{a,i}$ (making the messages to be
    conveyed and the prevailing state sequence $s^i$
    implicit).

    Roughly speaking, our construction will guarantee that, subject to
    some technicalities (see \eqref{eq:NoFail} ahead), $n$ of the
    $\bar{n}$ pairs
    $(x_{a,1},s_{1}), \ldots, (x_{a,\bar{n}},s_{\bar{n}})$ will be a
    permutation (determined by $s^{\bar{n}}$ and $\tilde{s}^{n}$ and
    hence common to all encoders and decoders) of the $n$-tuple
    $(\tilde{x}_{a,1},\tilde{s}_{1}), \ldots,
    (\tilde{x}_{a,n},\tilde{s}_{n})$ so that the performance of the
    causal scheme will be essentially as good as that of the noncausal
    scheme (because the network law is memoryless conditional on the
    state sequence and hence permutation invariant.)

     \newcommand{\gi}{\kappa}
     The causal encoders and the decoders % $a$ observes
     observe the first state $s_{1}$ and look for
     the first time-index at which $\tilde{s}_{i}$ is equal to
     it. If none is found, they set $\gi(1)$ to zero and Encoder~$a$ sets $x_{a,1}$
     to some arbitrary symbol. Otherwise, they set $\gi(1)$ to be
     that time-index (so
     $\gi(1) = \min \{i \geq 1\colon \tilde{s}_{i} = s_{1}\}$), and Encoder~$a$
     % Encoder~$a$
     produces the symbol $\tilde{x}_{a,\gi(1)}$. Thereafter, the
     encoders and decoders mark the time-index $\gi(1)$ as ``used.''

     At the second time instance, they observe $s_{2}$ and
     search for the first unused time index at which $\tilde{s}_{i}$
     is equal to it. Again, if none is found, they set $\gi(2)$ to zero
     and Encoder~$a$ sets $x_{a,2}$ to some arbitrary
     symbol. Otherwise, they set
     $\gi(2)$ to equal that time-index (so
     $\gi(2) = \min \{i \neq \gi(1)\colon \tilde{s}_{i} = s_{2}\}$)
     and Encoder~$a$ produces the symbol $\tilde{x}_{a,\gi(2)}$. Thereafter,
     all encoders and decoders also mark the time-index $\gi(2)$ as
     ``used.''  We continue in this fashion $\bar{n}$ times. By then,
     all the time indices $1, \ldots, n$ will have been marked ``used''
     % transmission at time $n$
     provided that
     \begin{IEEEeqnarray}{rCl}
       N(s\,|\,s^{\bar{n}}) & \geq & N(s\,|\,\tilde{s}^n), \quad \forall s
       \in \mathcal{S} \label{eq:NoFail}
     \end{IEEEeqnarray}
     which holds with probability tending to $1$ as $n \to \infty$
     by~\eqref{eq:WLLN}.  (If~\eqref{eq:NoFail} does not hold, our
     causal scheme fails and produces an error.)
%     The transmissions
%     from time $n+1$ through time $\bar{n}$ are arbitrary and the
%     corresponding outputs are ignored by the decoders.

     The receivers---knowing $s^{\bar{n}}$ (by our assumption that
     they are cognizant of the state) and knowing $\tilde{s}^{n}$
     (which was fixed ahead of time)---can recover the mapping
     $\gi(\cdot)$. Since no two elements of $[1:\bar{n}]$ are mapped
     to the same \emph{nonzero} element of $[1:n]$, we can define the
     reverse mapping $\gi^{-1}\colon [1:n] \to [1:\bar{n}]$

 %    And since the latter is a one-to-one mapping from
 %    $[1:n]$ to $\{\gi(i)\}_{i=1}^{n}$, they can invert it, i.e.,
 %    compute $i \in [1:n]$ from $\gi(i)$, an operation we denote
 %    $\gi^{-1}(i)$.
     Each decoder now rearranges its received sequence at times $1$
     through $n$, with Decoder~$b$ rearranging the received sequence
     $y_{b}^{\bar{n}}$ to obtain the sequence
     $y_{b,\gi^{-1}(1)}, \ldots, y_{b,\gi^{-1}(n)}$ and feeds this
     latter sequence to its noncausal decoder counterpart to produce
     the guess
     \begin{equation*}
       \phi_{b}\bigl( y_{b,\gi^{-1}(1)}, \ldots, y_{b,\gi^{-1}(n)}, \tilde{s}_{1},
       \ldots, \tilde{s}_{n} \bigr).
     \end{equation*}
     % which can also be expressed as
     % \begin{equation*}
     %   \phi_{b}\bigl( y_{b,\gi^{-1}(1)}, \ldots, y_{b,\gi^{-1}(n)}, s_{\gi(1)},
     %   \ldots, s_{\gi(n)} \bigr).
     % \end{equation*}
     For large enough $n$, the probability of \eqref{eq:NoFail}
     exceeds $1-p$, which, together with \eqref{eq:ChoiceOftilde}
     implies that $\mathrm{Pr}_{\textnormal{c}}\{\mathcal{E}\} < 3p$.

     A more formal account follows.  Let the sequence $\tilde{s}^{n}$ satisfy
     \eqref{sub:Choicetildes}, and let $\tilde{P}_S$ denote its
     type, so
     \begin{equation}
     \tilde{P}_S(s)=N(s\,|\,\tilde{s}^n)/n,
     \end{equation}
     where $N(s|\tilde{s}^n)$ denotes the number of occurrences of $s$
     in $\tilde{s}^n$.

     We reorder the channel uses to group the same state realizations in $\tilde{s}^n$ together. That is, for every $i\in[1:n]$, record the state at Time~$i$ and the number of times this state has occurred up to Time $i$:
     \begin{equation}
     \tilde{g}(i) = \bigl(\tilde{s}_i,N(\tilde{s}_i\,|\,\tilde{s}^i)\bigr).
     \end{equation}
     The mapping $\tilde{g}$ is invertible (because, as we recall, $\tilde{s}^n$ is fixed). Therefore,
  the encoding function employed by Encoder~$a$, conditional on $S^n=\tilde{s}^n$, can be equivalently expressed using the following mappings:
  \begin{equation}
  \left\{\tilde{f}_{a,s,j} \right\}_{s\in\set{S},\, j\in[1:n\tilde{P}_S(s)]},
  \end{equation}
  where
\begin{equation}\label{eq:fas}
\tilde{f}_{a,s,j} \bigl( \{M_\sigma\}_{\sigma\in\set{I}_a} \bigr)= f_{a,\,\tilde{g}^{-1}(s,j)}^{\textnormal{nc}} \bigl( \{M_\sigma\}_{\sigma\in\set{I}_a} , \tilde{s}^n \bigr),
\end{equation}
with the right-hand side (RHS) defined in \eqref{eq:nc_enc_separate}. The decoding function of Decoder~$b$ can also be expressed in terms of the new indices:
\begin{equation}
\tilde{\phi}_b \left( \{ y_{b,s,j}\}_{s\in\set{S},\, j\in[1:n\tilde{P}_S(s)]} \right) = \phi_b \left( \{y_{b,\, \tilde{g}^{-1}(s,j)}\}, \tilde{s}^n \right).
\end{equation}

We next describe a causal encoding strategy for $(1+2\delta)n$ channel uses. Encoder~$a$ is given by the following mappings:
\begin{equation}
f_{a,i}^{\textnormal{c}}\left( \{M_\sigma\}_{\sigma\in\set{I}_\sigma}, s^i \right) = \tilde{f}_{a, s_i, N(s_i\,|\,s^i)} \left( \{M_\sigma\}_{\sigma\in\set{I}_\sigma} \right),
\end{equation}
if $N(s_i\,|\,s^i) \le n\tilde{P}_S(s_i)$. Otherwise, pick an input at random.

We now turn to the decoders. Denote
\begin{equation}
\set{A} \triangleq \left\{ \textnormal{every $s\in\set{S}$ appears at least $n\tilde{P}_S(s)$ times in $S^{(1+2\delta)n}$}\right\}.
\end{equation}
By \eqref{eq:WLLN} and since $\tilde{P}_S$ is close to $P_S$, for sufficiently large $n$, the probability for $\set{A}$ to happen is at least $1-p$.
If $\set{A}$ is false, then all decoders declare an error. If $\set{A}$ is true, then all decoders reorder the channel uses to group the same state realizations together as above. This option is characterized by the mapping $g$:
\begin{equation}
g(i, s^n) = \bigl(s,N(s_i\,|\,s^i)\bigr).
\end{equation}
For every $s\in\set{S}$, the decoders keep the first $n \tilde{P}_S(s)$ channel uses where the state equals $s$ and discard the rest. On the channel outputs that they keep, Decoder $b$ applies $\tilde{\phi}_b$ to recover its desired messages. That is,
\begin{equation}
\phi_b \left( y_b^n,s^n \right) = \tilde{\phi}_b \bigl( \{y_{b,\,g(i,s^n)} \}_{i\in[1:n]}\bigr).
\end{equation}

Provided that $\set{A}$ is true, the causal scheme described above has exactly the same error probability as the given noncausal scheme, when the latter is conditional on $S^n=\tilde{s}^n$, i.e.,
\begin{equation}
\mathrm{Pr}_{\textnormal{c}} \{\set{E}\,|\,\set{A}\} = \mathrm{Pr}_{\textnormal{nc}}\left\{\mathcal{E}\,\middle|\, S^n=\tilde{s}^n \right\} .
\end{equation} 

Indeed, for any $\{M_\sigma\}_{\sigma\in[1:\Sigma]}$, after reordering, the channel inputs in the causal case are exactly the same as those in the noncausal case specified to $S^n=\tilde{s}^n$, following the construction in \eqref{eq:fas}. Since the channel is memoryless given the states, the joint distribution of the output symbols (again after reordering) is also the same between the two cases. We can thus bound the error probability of the causal coding scheme as follows:
    \begin{IEEEeqnarray}{rCl}
%    \IEEEyesnumber\IEEEyessubnumber*
%      \IEEEeqnarraymulticol{3}{l}{
       \mathrm{Pr}_{\textnormal{c}}\{\set{E}\}
%      }\nonumber\\* \quad
       & = & \mathrm{Pr}_{\textnormal{c}} \{\set{E}\,|\,\set{A}\} \cdot \Pr \{ \set{A} \} +  \mathrm{Pr}_{\textnormal{c}} \{\set{E}\,|\,\set{A}^\textnormal{c}\} \cdot \Pr \{ \set{A}^\textnormal{c} \} \\
       & \le &  \mathrm{Pr}_{\textnormal{c}} \{\set{E}\,|\,\set{A}\} + \Pr \{ \set{A}^\textnormal{c} \}\\
       & = &  \mathrm{Pr}_{\textnormal{nc}}\left\{\mathcal{E}\,\middle|\, S^n=\tilde{s}^n \right\} + \Pr \{ \set{A}^\textnormal{c} \} \\
       & \le & 2p  + p =3p,
    \end{IEEEeqnarray}
%    where the inequality is true when $n>\max (N_0,N_1)$.
establishing the desired bound.
\end{IEEEproof}

\section{Acknowledgments}
This work was supported by the Swiss National
Science Foundation (SNSF) under Grant 200021-215090.

\end{document}